\documentclass[conference]{IEEEtran}
\usepackage[utf8]{inputenc} 
\usepackage[T1]{fontenc}    
\usepackage{hyperref}       
\usepackage{url}            
\usepackage{booktabs}       
\usepackage{amsfonts}       
\usepackage{nicefrac}       
\usepackage{microtype}      
\usepackage{amsmath}
\usepackage{theorem}
\theorembodyfont{\upshape}
\usepackage{graphicx}
\usepackage{bm}
\def\qed{\hfill $\Box$}

\newtheorem{defi}{Definition}
\newtheorem{theo}{Theorem}
\newtheorem{proof}{Proof}
\newtheorem{example}{Example}
\DeclareMathOperator*{\argmin}{arg\,min}
\DeclareMathOperator*{\argmax}{arg\,max}

\ifCLASSINFOpdf
\else
\fi
\begin{document}
%
\title{A Note on the Estimation Method of Intervention Effects based on Statistical Decision Theory}

\author{\IEEEauthorblockN{Shunsuke Horii}
\IEEEauthorblockA{Waseda University\\
1-6-1, Nishiwaseda, Shinjuku-ku,\\
Tokyo 169-8050, Japan\\
Email: s.horii@aoni.waseda.jp}
\and
\IEEEauthorblockN{Tota Suko}
\IEEEauthorblockA{Waseda University\\
1-6-1, Nishiwaseda, Shinjuku-ku,\\
Tokyo 169-8050, Japan\\
Email: suko@waseda.jp}
}


%


\maketitle

\begin{abstract}
In this paper, we deal with the problem of estimating the intervention effect in the statistical causal analysis using the structural equation model and the causal diagram.
 The intervention effect is defined as a causal effect on the response variable $Y$ when the causal variable $X$ is fixed to a certain value by an external operation and is defined based on the causal diagram.
 The intervention effect is defined as a function of the probability distributions in the causal diagram, however, generally these probability distributions are unknown, so it is required to estimate them from data.
In other words, the steps of the estimation of the intervention effect using the causal diagram are as follows: 1. Estimate the causal diagram from the data, 2. Estimate the probability distributions in the causal diagram from the data, 3. Calculate the intervention effect.
However, if the problem of estimating the intervention effect is formulated in the statistical decision theory framework, estimation with this procedure is not necessarily optimal. In this study, we formulate the problem of estimating the intervention effect for the two cases, the case where the causal diagram is known and the case where it is unknown, in the framework of statistical decision theory and derive the optimal decision method under the Bayesian criterion. We show the effectiveness of the proposed method through numerical
simulations.
\end{abstract}


%
\IEEEpeerreviewmaketitle

\section{Introduction}

Causal analysis based on linear structural equation model and path analysis is widely used in sociology, economics, biology, etc.
Pearl extended the concept of total effects in the path analysis to a general structural equation model and defined it as the intervention effect \cite{pearl1995causal}.
Fixing a variable $X$ at a certain value $x$ by an external operation is called intervention, and the intervention effect is mathematically defined as a causal effect on the response variable $Y$.
The intervention effect is defined based on a causal diagram that expresses the existence or nonexistence of a causal relationship between variables and conditional probability distributions that expresses causal relationships among variables.
However, in general, the causal diagram and the conditional probability distributions among variables are unknown, so it is necessary to estimate both from the data.
That is, the calculation of the intervention effect based on the causal diagram consists of the following steps.
\begin{enumerate}
 \item Estimate a causal diagram from the data
 \item Estimate the conditional probability distributions among variables from the data
 \item Calculate the intervention effect
\end{enumerate}
The estimation methods of the causal diagram are roughly divided into two categories: constraint-based methods (such as PC algorithm \cite{spirtes1991algorithm}) that estimates the structure with constraints such as conditional independence among variables, and score-based methods (such as GES algorithm \cite{chickering2002finding}) that output a graph with the maximum approximate value of posterior probability.
Estimation of a conditional probability distribution is a general topic not limited to causal inference, and widely used approaches are estimating a parameter by assuming a parametric probability distribution or estimating by a nonparametric method.
In this research, we assume parametric probability distributions for the conditional probability distributions.
Although it is known that the identifiability of causal diagrams would change by assumptions on the conditional probability distributions \cite{shimizu2006linear}, this research does not deal with that point in depth.
However, we note that the proposal in this research is applicable as long as parametric distribution is assumed for the conditional probability distribution.
Since the intervention effect is defined on the causal diagram and the conditional probability distributions, it seems natural to estimate it by the above procedure.
However, if we formulate the problem of estimating the intervention effect based on the statistical decision theory, estimating it by this procedure is not necessarily optimal.
In this study, the problem of estimating the intervention effect is formulated in the framework of the statistical decision theory for each case where the causal diagram is known and unknown, and the optimal decision function is derived under the Bayes criterion.
The remainder of the paper is organized as follows.
In Section 2, the definitions of the structural equation model, causal diagram, and intervention effect are described.
In Section 3, we formulate the problem to estimate the intervention effect as a statistical decision problem for the case where the causal diagram is known and derive the optimal decision function under the Bayes criterion.
In Section 4, we do the same thing as in Section 3 for the case where the causal diagram is unknown.
In Section 5, we evaluate the effectiveness of the proposed method by comparing the intervention effect estimated by the proposed method and that estimated by two stage method, that is, calculate the intervention effect after estimating the causal diagram and/or the conditional probability distributions.
Finally, we give a summary and future works in Section 6.

\section{Causal diagram and intervention effect}

Here, after describing the definition of the causal diagram, we describe the mathematical definition of the intervention effect.

\subsection{Causal diagram}
\begin{defi}
 Let $G$ be a directed acyclic graph (DAG) and $V=(X_{1}, X_{2}, \ldots, X_{m})$ be a set of random variables that corresponds to the set of the vertices of $G$.
$G$ is called a causal diagram if it specifies the causal relationships among variables in the following form,
\begin{align}
X_{i}=g_{i}(\mbox{pa}(X_{i}),\epsilon_{i}),\quad i=1,\ldots,m, \label{SEM}
\end{align}
and the random variables are generated according to this causal relationship.
The equations (\ref{SEM}) are called structural equations for $X_{1}, X_{2}, \ldots, X_{m}$.
$\mbox{pa}(X_{i})\subset V$ is the set of variables that have an arrow that heads to $X_{i}$.
We assume that $\epsilon_{1}, \epsilon_{2}, \ldots, \epsilon_{m}$ are mutually independent.
\end{defi}
Let $p(x_{i}|\mbox{pa}(x_{i}))$ be the conditional probability distribution of $X_{i}$ given $\mbox{pa}(X_{i})$.

\begin{figure}[t]
\begin{center}
\includegraphics[keepaspectratio=true,width=\linewidth]{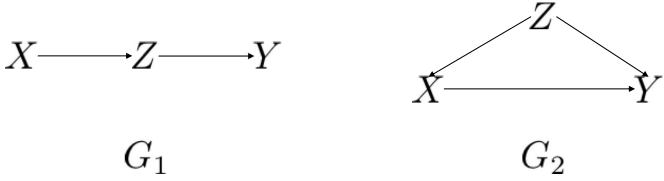}
\caption{Examples of causal diagram.}
\label{fig_diagram}
\end{center}
\end{figure}

\begin{example}
If the causal diagram of the random variables $X, Y, Z$ is $G_{1}$ in Figure \ref{fig_diagram}, there are causal relationships,
\begin{align}
Z&=g_{Z}(X, \epsilon_{Z}),\\
Y&=g_{Y}(Z, \epsilon_{Y}).
\end{align}
Similarly, if the causal diagram of the random variables $X, Y, Z$ is $G_{2}$ in Figure \ref{fig_diagram}, there are causal relationships,
\begin{align}
X&=g_{X}(Z, \epsilon_{X}),\\
Y&=g_{Y}(X, Z, \epsilon_{Y}).
\end{align}
\end{example}

\subsection{Intervention effect}
In a causal diagram, an external operation that fixes the value of $X$ to a constant regardless of the value of other variables is called intervention, and the distribution of $Y$ after the intervention is called intervention effect. Its mathematical definition is given as follows \cite{pearl1995causal}.
\begin{defi}
Let $V=\left\{X, Y, Z_{1}, Z_{2}, \ldots, Z_{p}\right\}$ be the set of vertices of a causal diagram $G$.
The intervention on $Y$ when intervening $X=x$ is defined as 
\begin{align}
 p(y|\mbox{do}(X=x))=\int\cdots\int \frac{p(x,y,z_{1},\ldots,z_{p})}{p(x|\mbox{pa}(x))}dz_{1}\ldots dz_{p}\label{effect_no_parameter}.
\end{align}
$\mbox{do}(X=x)$ means that $X$ is fixed to $x$ by intervention.
\end{defi}

(\ref{effect_no_parameter}) can be calculated only after the causal diagram is determined and the conditional distributions among the random variables are estimated.
Let $m$ be the variable that represents the causal diagram and the conditional probability distributions are parametric distributions specified by a parameter $\bm{\theta}_{m}$.
To clarify that the intervention effect depends on $m$ and $\bm{\theta}_{m}$, we rewrite (\ref{effect_no_parameter}) as follows.
\begin{multline}
p(y|\mbox{do}(X=x),m,\bm{\theta}_{m})=\\
\int\cdots\int \frac{p(x,y,z_{1},\ldots,z_{p}|m, \bm{\theta}_{m})}{p(x|\mbox{pa}(x),m, \bm{\theta}_{m})}dz_{1}\ldots dz_{p}.\label{effect_parameter}
\end{multline}

\begin{example}
 Assume that the causal diagram $m$ of $X, Y, Z$ is $G_{1}$ in Figure \ref{fig_diagram} and the structural equations are linear, that is,
\begin{align}
 Z&=\theta_{Z|X}X+\epsilon_{Z},\quad \epsilon_{Z}\sim \mathcal{N}(0, 1^{2}),\label{example_SEM_1_1}\\
Y&=\theta_{Y|Z}Z+\epsilon_{Y},\quad \epsilon_{Y}\sim\mathcal{N}(0, 1^{2}),\label{example_SEM_1_2}
\end{align}
where $\mathcal{N}(\mu, \sigma^{2})$ denotes the normal distribution with mean $\mu$ and variance $\sigma^{2}$.
Then, $\bm{\theta}_{m}=(\theta_{Z|X}, \theta_{Y|Z})$ and the intervention effect on $Y$ when intervening $X=x$ is given by
\begin{align}
 p(y|\mbox{do}(X=x),m=G_{1}, \bm{\theta}_{m})=\mathcal{N}(y; \theta_{Y|Z}\theta_{Z|X}x, 1+\theta_{Y|Z}^{2}),
\end{align}
where $\mathcal{N}(\cdot; \mu, \sigma^{2})$ denotes the probability density function of $\mathcal{N}(\mu, \sigma^{2})$.
In this case, it is well known that the intervention effect equals to the conditional pribability distribution $p(y|x,\bm{\theta}_{m})$ and the above formula describes this in detail.

Similarlly, assume that the causal diagram $m$ of $X, Y, Z$ is $G_{2}$ in Figure \ref{fig_diagram} and the structural equations are given by
\begin{align}
X&=\theta_{X|Z}Z+\epsilon_{X},\quad \epsilon_{X}\sim\mathcal{N}(0, 1^{2}),\label{example_SEM_2_1}\\
Y&=\theta_{Y|X}X+\theta_{Y|Z}Z+\epsilon_{Y},\quad \epsilon_{Y}\sim\mathcal{N}(0, 1^{2}).\label{example_SEM_2_2}
\end{align}
Then, $\bm{\theta_{m}}=(\theta_{X|Z}, \theta_{Y|X}, \theta_{Y|Z})$ and the intervention effect on $Y$ when intervening $X=x$ is given by
\begin{align}
&p(y|\mbox{do}(X=x),m=G_{2}, \bm{\theta}_{m})=\mathcal{N}(y; \tilde{\mu}, \tilde{s}^{-1}),\\
&\tilde{\mu}=\tilde{s}^{-1}\theta_{Y|X}x-\mu_{Z}\theta_{Y|Z},\\
&\tilde{s}=\frac{s_{Z}}{\theta_{Y|Z}^{2}+s_{Z}},
\end{align}
where we assumed that $Z\sim\mathcal{N}(\mu_{Z}, s_{Z}^{-1})$.
\end{example}

\section{Decision theoretic approach for estimating intervention effect; causal diagram is known}
Here, we consider the case where the causal diagram $m$ is known, but $\bm{\theta}_{m}$ is unknown.
In this case, we cannot calculate (\ref{effect_parameter}) directly and we have to estimate it from the data.
Let $D^{n}=(x_{n}, y_{n}, z_{1n},\ldots, z_{pn})_{n=1,\ldots,N}$ be a sample of $X, Y, Z_{1}, \ldots, Z_{p}$ with size $n$.
Decision function $AP:D^{n}\mapsto p(y|x)$ outputs an estimate of the intervention effect.
We have to define some loss function for the decision function.
In this study, the Kullback-Leibler divergence with the intervention effect is used as a loss function.
\begin{multline}
Loss(\bm{\theta}_{m}, AP(D^{n}))= \\
\int p(y|\mbox{do}(X=x),m,\bm{\theta}_{m})\ln \frac{p(y|\mbox{do}(X=x),m,\bm{\theta}_{m})}{AP(D^{n})(y|x)}dy.\label{loss}
\end{multline}
The risk function is defined as the expectation of the loss function with respect to $D^{n}$.
\begin{align}
Risk(\bm{\theta}_{m},AP)=E_{D^{n}|\bm{\theta}}\left[Loss(\bm{\theta}_{m}, AP(D^{n}))\right].
\end{align}
The risk function is a function of the parameter $\bm{\theta}_{m}$ and there is no decision function that minimizes the risk function for all parameter $\bm{\theta}_{m}\in\Theta_{m}$.
In this study, we assume a prior distribution $p(\bm{\theta}_{m})$ for the parameter $\bm{\theta}_{m}$ and consider the following Bayes risk function.
\begin{align}
BR(AP)=E_{\bm{\theta}_{m}}\left[Risk(\bm{\theta}_{m}, AP)\right].\label{BR}
\end{align}
Then, the following theorem holds.
\begin{theo}
 \label{theorem1}
The Bayes optimal decision function that minimizes (\ref{BR}) is given by 
\begin{align}
AP^{*}(D^{n})=p(y|\mbox{do}(X=x),m,D^{n}),\label{bayes_optimal_fixed_model}
\end{align}
where 
\begin{multline}
p(y|\mbox{do}(X=x),m,D^{n})=\\
\int p(y|\mbox{do}(X=x),m,\bm{\theta}_{m})p(\bm{\theta}_{m}|m, D^{n})d\bm{\theta}_{m},\label{predict_fixed_model}
\end{multline}
\end{theo}

\begin{proof}
 The minimization of the Bayes risk function is reduced to the minimization of the loss function weighted by the posterior distribution \cite{berger2013statistical}.
That is, 
\begin{multline}
\argmin_{AP} BR(AP)=\\
\argmin_{AP} \int Loss(\bm{\theta}_{m}, AP(D^{n}))p(\bm{\theta}_{m}|m, D^{n})d\bm{\theta}_{m}.
\end{multline}
Substituting (\ref{loss}) into the loss function and removing the terms that do not depend on $AP$, we have
\begin{align}
\argmin_{AP} BR(AP)=\argmax_{AP} \int\int p(y|\mbox{do}(X=x), m,\bm{\theta}_{m})\nonumber\\
\times p(\bm{\theta}_{m}|m,D^{n})\ln AP(D^{n})d\bm{\theta}_{m}dy\\
=\argmax_{AP}\int p(y|\mbox{do}(X=x),m,D^{n})\ln AP(D^{n})dy.
\end{align}
From Shannon's inequality \cite{cover2012elements},
\begin{multline}
 \argmax_{AP}\int p(y|\mbox{do}(X=x),m,D^{n})\ln AP(D^{n})dy=\\
 p(y|\mbox{do}(X=x),m,D^{n}).
\end{multline}
\qed
\end{proof}

\begin{example}
 Assume that the causal diagram $m$ for $X, Y, Z$ is $G_{1}$ in Figure \ref{fig_diagram} and the structural equations are given by (\ref{example_SEM_1_1}) and (\ref{example_SEM_1_2}).
In addition, as the prior distributions of $\theta_{Y|Z}, \theta_{Z|X}$, assume that $\theta_{Y|Z}, \theta_{Z|X}\sim\mathcal{N}(0,\alpha^{-1})$.
Then, the Bayes optimal estimator of the intervention effect is given by
\begin{align}
p(y|\mbox{do}(X=x), m=G_{1}, D^{n})=\nonumber\\
\int\int \mathcal{N}(y; \theta_{Y|Z}\theta_{Z|X}x, 1+\theta_{Y|Z}^{2})\mathcal{N}(\theta_{Y|Z}; \mu_{Y|Z}, s_{Y|Z}^{-1})\times \nonumber \\
\mathcal{N}(\theta_{Z|X}; \mu_{Z|X}, s_{Z|X}^{-1})d\theta_{Y|Z}d\theta_{Z|X},\label{predict_example_1}
\end{align}
\begin{align}
\mu_{Y|Z}&=s_{Y|Z}^{-1}\bm{z}^{T}\bm{y},\\
s_{Y|Z}&=\alpha+\bm{z}^{T}\bm{z},\\
\mu_{Z|X}&=s_{Z|X}^{-1}\bm{x}^{T}\bm{z},\\
s_{Z|X}&=\alpha+\bm{x}^{T}\bm{x},
\end{align}
where $\bm{x}=(x_{1},\ldots,x_{N})^{T}, \bm{y}=(y_{1},\ldots,y_{N})^{T}, \bm{z}=(z_{1},\ldots, z_{N})$.

Similarly, assume that the causal diagram $m$ for $X, Y, Z$ is $G_{2}$ in Figure \ref{fig_diagram} and the structural equations are given by (\ref{example_SEM_2_1}) and (\ref{example_SEM_2_2}).
In addition, as the prior distributions of $\theta_{Y|X}, \theta_{Y|Z}$, assume that $\theta_{Y|X}, \theta_{Y|Z}\sim\mathcal{N}(0, \alpha^{-1})$.
Let $\bm{\theta}_{Y|XZ}=(\theta_{Y|X},\theta_{Y|Z})$, then, the Bayes optimal estimator of the intervention effect is given by
\begin{multline}
p(y|\mbox{do}(X=x), m=G_{2}, D^{n})=\\
\int \mathcal{N}(y; \tilde{\mu}, \tilde{s}^{-1})\mathcal{N}(\bm{\theta}_{Y|XZ};\bm{\mu}_{Y|XZ}, \bm{S}_{Y|XZ}^{-1})d\bm{\theta}_{Y|XZ},\label{predict_example_2}
\end{multline}
\begin{align}
\tilde{\mu}&=\tilde{s}^{-1}\theta_{Y|X}x-\mu_{Z}\theta_{Y|Z}\\
\tilde{s}&=\frac{\alpha s_{Z}}{\alpha\theta_{Y|Z}^{2}+s_{Z}}\\
\bm{\mu}_{Y|XZ}&=\bm{S}_{Y|XZ}^{-1}\bm{X}_{\setminus \bm{y}}^{T}\bm{y},\\
\bm{S}_{Y|XZ}&=\alpha\bm{I}+\bm{X}_{\setminus \bm{y}}^{T}\bm{X}_{\setminus \bm{y}},
\end{align}
where $\mathcal{N}(\cdot; \bm{\mu}, \bm{\Sigma})$ denotes the probability density function of the mulrivariate normal distribution with mean vector $\bm{\mu}$ and covariance matrix $\bm{\Sigma}$ and 
\begin{align}
\bm{X}_{\setminus \bm{y}}=
\begin{pmatrix}
\bm{x}^{T}\\
\bm{z}^{T}
\end{pmatrix}^{T}.
\end{align}

We note that the Bayes optimal estimator (\ref{predict_example_1}) and (\ref{predict_example_2}) cannot be calculated analytically even in the cases of the linear structural equation model of these examples.
In the later experiments, we performed a numerical integration for the calculations.
\end{example}

\section{Decision theoretic approach for estimating intervention effect; causal diagram is unknown}

Here, we consider the case where not only the parameter $\bm{\theta}_{m}$, but also the causal diagram $m$ is unknown.
Since $m$ is unknown, the loss function is defined for $m$ and $\bm{\theta}_{m}$.
\begin{multline}
Loss(m,\bm{\theta}_{m}, AP(D^{n}))=\\
\int p(y|\mbox{do}(X=x),m,\bm{\theta}_{m})\ln \frac{p(y|\mbox{do}(X=x),m,\bm{\theta}_{m})}{AP(D^{n})(y|x)}dy.
\end{multline}
The risk function is given by
\begin{align}
Risk(m,\bm{\theta}_{m},AP)=E_{D^{n}|\bm{\theta}_{m},m}\left[Loss(m, \bm{\theta}_{m}, AP(D^{n}))\right].
\end{align}
In this study, we consider the case where the set of candidate causal diagrams is given by $\mathcal{M}$ and we can assume the prior distribution $p(m)$ for $m\in\mathcal{M}$ and $p(\bm{\theta}_{m}|m)$ for $\bm{\theta}_{m}$ under $m$.
Then, the Bayes risk function is given by
\begin{align}
BR(AP)=E_{m}\left[E_{\bm{\theta}_{m}|m}\left[Risk(m,\bm{\theta}_{m}, AP)\right]\right].\label{BR_model}
\end{align}
In this case, the following theorem holds.
\begin{theo}
 The Bayes optimal estimator that minimizes (\ref{BR_model}) is given by 
\begin{align}
AP^{*}(D^{n})=p(y|\mbox{do}(X=x),D^{n}),\label{predict_mixed_model}
\end{align}
where 
\begin{multline}
p(y|\mbox{do}(X=x),D^{n})=\\
\sum_{m\in\mathcal{M}}p(m|D^{n})p(y|\mbox{do}(X=x),m,D^{n}),
\end{multline}
and $p(y|\mbox{do}(X=x),m,D^{n})$ is given by (\ref{predict_fixed_model}).
\end{theo}
\begin{proof}
 It is proved in the same manner as the proof of Theorem 1.\qed
\end{proof}

\begin{example}
Assume that the set $\mathcal{M}$ of the candidate causal diagrams is $\left\{G_{1}, G_{2}\right\}$ in Figure \ref{fig_diagram} and the structural equations under each causal diagram are given in the same way as in Examples 2 and 3.
When the prior distribution of the model $m$ is $p(m=G_{1}), p(m=G_{2})$ and the prior distribution of the parameter $\bm{\theta}_{m}$ under each model are given in the same way as in Example 3, the Bayes optimal estimator of the intervention effect is given by
\begin{align}
&p(y|\mbox{do}(X=x),D^{n})=\nonumber\\
&p(m_{1}|D^{n})p(y|\mbox{do}(X=x),m=G_{1},D^{n})+\\
&p(m_{2}|D^{n})p(y|\mbox{do}(X=x),m=G_{2},D^{n}),\nonumber
\end{align}
where $p(y|\mbox{do}(X=x),m=G_{1},D^{n}),p(y|\mbox{do}(X=x),m=G_{2},D^{n})$ are the same as given by (\ref{predict_example_1}) (\ref{predict_example_2}).
\end{example}

\section{Numerical experiments}
In this section, we show the effectiveness of the proposed method through numerical simulations.

\subsection{Case 1 : causal diagram is known}
\label{experiment_1}
First, we deal with the case where the causal diagram is known.
We consider the two cases, one is that the true diagram is $G_{1}$ in Figure \ref{fig_diagram} and the other is that the true diagram is $G_{2}$ in Figure \ref{fig_diagram}.
The structural equations are (\ref{example_SEM_1_1}) and (\ref{example_SEM_1_2}) for $G_{1}$ and (\ref{example_SEM_2_1}) and (\ref{example_SEM_2_2}) for $G_{2}$.
We assume that the probability distributions of variables corresponding to leaf nodes in each model, that is, $X$ in $G_{1}$ and $Z$ in $G_{2}$, are both $\mathcal{N}(0, 1^{2})$.
We also assume that the prior distributions of the parameters under each model, that is, $\theta_{Y|Z}, \theta_{Z|X}$ in $G_{1}$ and $\theta_{X|Z}, \theta_{Y|X}, \theta_{Y|Z}$ in $G_{2}$, are all $\mathcal{N}(0, 1^{2})$.
We consider the problem to estimate the intervention effect on $Y$ when intervening $X=1$ given $D^{n}=(x_{n}, y_{n}, z_{n})_{n=1,\ldots,N}$ as a sample of $(X, Y, Z)$.
We compare the following three methods.
\begin{description}
 \item[Method 1 (ML)] \ \\ Calculate the maximum likelihood (ML) estimator $\bm{\theta}_{m, ML}$ by
\begin{align}
  \hat{\bm{\theta}}_{m, ML}=\argmax_{\bm{\theta}_{m}}p(D^{n}|\bm{\theta}_{m}),
\end{align}
and substitute it to (\ref{effect_parameter}).
 \item[Method 2 (MAP)] \ \\Calculate the maximum a posteriori (MAP) estimator $\bm{\theta}_{m, MAP}$ by
\begin{align}
 \hat{\bm{\theta}}_{m,MAP}=\argmax_{\bm{\theta}_{m}}p(\bm{\theta}_{m}|D^{n}),
\end{align}
and substitute it to (\ref{effect_parameter}).
 \item[Method 3 (BAYES)] \ \\Calculate the Bayes optimal estimator (\ref{bayes_optimal_fixed_model}).
\end{description}

Figure \ref{fig_result1} shows the Kullback-Leibler divergence between the true intervention effect on $Y$ when intervening $X=1$ in the model $G_{1}$ and the estimator of each method.
Figure \ref{fig_result2} is the same result for the model $G_{2}$.
In either case, as the sample size increases, the results of the three methods converge.
This can be explained by the fact that the posterior distribution of parameters concentrates around the MAP estimator as the sample size increases, and the MAP estimator and the ML estimator also approaches.
However, when the sample size is small, method 2 is better than method 1, and method 3 is better than method 2.
In this experiment, we experimented with models with very few variables, so the difference of each method is small, but it is expected that the difference of each method will become larger as the model becomes more complicated.

\begin{figure}[t]
 \begin{center}
  \includegraphics[keepaspectratio=true,width=\linewidth]{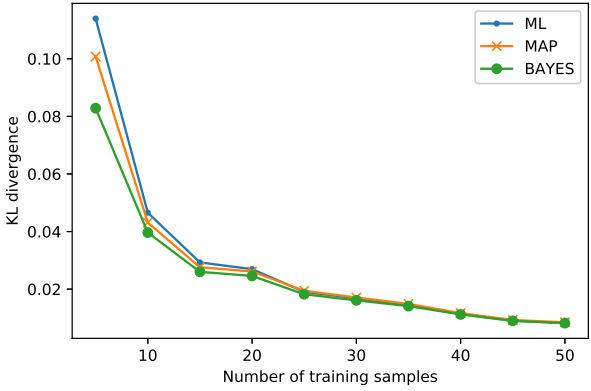}
\caption{The Kullback-Leibler divergence between the true intervention effect on $Y$ when intervening $X=1$ in the model $G_{1}$ and the estimator of each method.}
\label{fig_result1}
 \end{center}
\end{figure}

\begin{figure}[t]
\begin{center}
  \includegraphics[keepaspectratio=true,width=\linewidth]{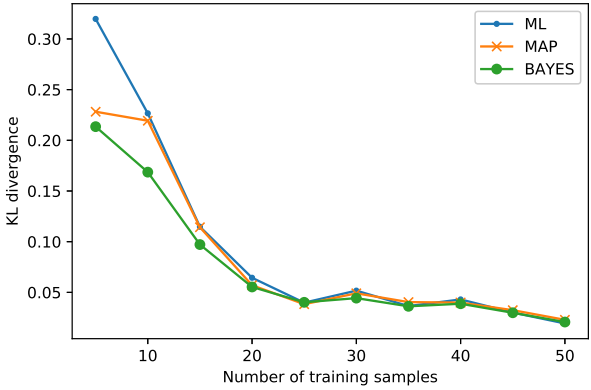}
\caption{The Kullback-Leibler divergence between the true intervention effect on $Y$ when intervening $X=1$ in the model $G_{2}$ and the estimator of each method.}
\label{fig_result2}
 \end{center}
\end{figure}

\subsection{Case 2 : causal diagram is unknown}
Next, we deal with the case where the causal diagram is unknown.
Let the set $\mathcal{M}$ of the candidates of the causal model be $\left\{G_{1}, G_{2}\right\}$ in Figure \ref{fig_diagram}.
The assumptions for the structural equations, the probability distributions of the leaf variables, and the prior distributions of the parameters are the same as the previous experiment.
We also assume that $p(m=G_{1})=p(m=G_{2})=\frac{1}{2}$.
Note that $X$ and $Y$ are conditionally independent when $Z$ is given in the model $G_{1}$, but they are not in the model $G_{2}$, so we can identify that which model generated data with high probability as the sample size increases.
As in the case of the previous experiment, we consider the problem to estimate the intervention effect on $Y$ when intervening $X=1$ given $D^{n}=(x_{n}, y_{n}, z_{n})_{n=1,\ldots,N}$ as a sample of $(X, Y, Z)$.
We compare the following two methods.
\begin{description}
 \item[Method 1 (MAP)] \ \\Estimate the model by
\begin{align}
 \hat{m}=\argmax_{m\in\mathcal{M}}p(m|D^{n})
\end{align}
and calculate the Bayes optimal estimator under the model $\hat{m}$,
\begin{align}
  p(y|\mbox{do}(X=x),\hat{m},D^{n}).
\end{align}
 \item[Method 2 (BAYES)] \ \\Calculate the Bayes optimal estimator (\ref{predict_mixed_model}).
\end{description}
Figure \ref{fig_result3} shows the Kullback-Leibler divergence between the true intervention effect on $Y$ when intervening $X=1$ and the estimator of each method.
The results of the two methods also approach as the sample size increases.
This can be explained from the fact that the as the sample size increases, the posterior probability of the true model approaches to $1$.
However, when the sample size is small, Method 2 is better than Method 1.
In this experiment, we experimented with only two candidate models, so there are differences between two methods only in small sample sizes.
It is expected that the difference will increase as the number of candidate models increases.

\begin{figure}[t]
 \begin{center}
  \includegraphics[keepaspectratio=true,width=\linewidth]{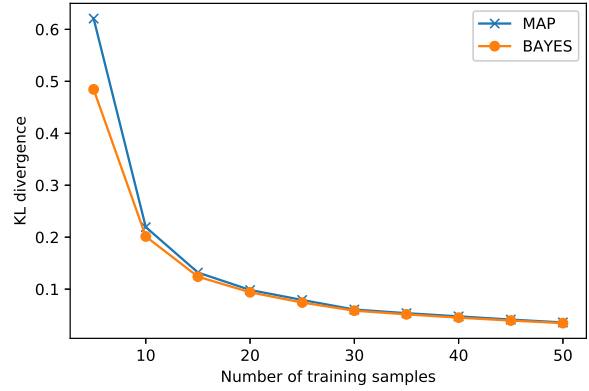}
\caption{The Kullback-Leibler divergence between the true intervention effect on $Y$ when intervening $X=1$ and the estimator of each method. Models $G_{1}$ and $G_{2}$ appear with equal probability.}
\label{fig_result3}
 \end{center}
\end{figure}

\section{Conclusion and future works}
In this study, the Bayes optimal estimation method for estimating the intervention effect was derived by formulating the estimation problem in the framework of the statistical decision theory.
In the estimation of the intervention effect, it is common to first estimate the causal diagram, estimate the conditional probability distributions among the variables, then calculate the intervention effect.
However, from the viewpoint of the Bayes decision theory framework, instead of determining models and parameters, weighting with a posterior probability or posterior distribution is optimal.

We describe some future works.
In the examples in this paper, we dealt with the case where the structural equations are linear.
It is necessary to derive the general form of the Bayes optimal estimator for those cases.
Further, it seems to be meaningful to investigate how the difference between the methods in the experiments becomes large in the cases other than the linear structural equation model.

In this study, we did not mention the calculation methods and computational complexity.
Even if the model is known and structural equations are linear, the Bayes optimal intervention effect estimator cannot be analytically calculated.
Therefore, in this paper, the estimator was calculated by numerical integration.
As the model becomes more complicated, the computational complexity will become higher.
It is necessary to construct an approximation algorithm that efficiently calculates the Bayes optimal estimator.
Also, when the model is unknown, it is necessary to calculate the posterior probability of all models, but as the number of candidate models becomes large, this also becomes computationally difficult.
It is also necessary to construct an approximation algorithm that efficiently calculates the Bayes optimal estimator in the case where the model is unknown.


\section*{Acknowledgment}

We would like to acknowledge all
members of Matsushima Lab. and Goto Lab. in Waseda Univ. for their
helpful suggestions to this work.
This research is partially supported by No. 16K00417 of Grant-in-Aid for
Scientific Research Category (C) and No. 18H03642 of Grant-in-Aid for Scientific Research Category (A), Japan Society for the Promotion
of Science.



%

\bibliographystyle{IEEEtran}
\bibliography{ref}

\begin{thebibliography}{1}
\providecommand{\url}[1]{#1}
\csname url@samestyle\endcsname
\providecommand{\newblock}{\relax}
\providecommand{\bibinfo}[2]{#2}
\providecommand{\BIBentrySTDinterwordspacing}{\spaceskip=0pt\relax}
\providecommand{\BIBentryALTinterwordstretchfactor}{4}
\providecommand{\BIBentryALTinterwordspacing}{\spaceskip=\fontdimen2\font plus
\BIBentryALTinterwordstretchfactor\fontdimen3\font minus
  \fontdimen4\font\relax}
\providecommand{\BIBforeignlanguage}[2]{{%
\expandafter\ifx\csname l@#1\endcsname\relax
\typeout{** WARNING: IEEEtran.bst: No hyphenation pattern has been}%
\typeout{** loaded for the language `#1'. Using the pattern for}%
\typeout{** the default language instead.}%
\else
\language=\csname l@#1\endcsname
\fi
#2}}
\providecommand{\BIBdecl}{\relax}
\BIBdecl

\bibitem{pearl1995causal}
J.~Pearl, ``Causal diagrams for empirical research,'' \emph{Biometrika},
  vol.~82, no.~4, pp. 669--688, 1995.

\bibitem{spirtes1991algorithm}
P.~Spirtes and C.~Glymour, ``An algorithm for fast recovery of sparse causal
  graphs,'' \emph{Social science computer review}, vol.~9, no.~1, pp. 62--72,
  1991.

\bibitem{chickering2002finding}
D.~M. Chickering and C.~Meek, ``Finding optimal bayesian networks,'' in
  \emph{Proceedings of the Eighteenth conference on Uncertainty in artificial
  intelligence}.\hskip 1em plus 0.5em minus 0.4em\relax Morgan Kaufmann
  Publishers Inc., 2002, pp. 94--102.

\bibitem{shimizu2006linear}
S.~Shimizu, P.~O. Hoyer, A.~Hyv{\"a}rinen, and A.~Kerminen, ``A linear
  non-gaussian acyclic model for causal discovery,'' \emph{Journal of Machine
  Learning Research}, vol.~7, no. Oct, pp. 2003--2030, 2006.

\bibitem{berger2013statistical}
J.~O. Berger, \emph{Statistical decision theory and Bayesian analysis}.\hskip
  1em plus 0.5em minus 0.4em\relax Springer Science \& Business Media, 2013.

\bibitem{cover2012elements}
T.~M. Cover and J.~A. Thomas, \emph{Elements of information theory}.\hskip 1em
  plus 0.5em minus 0.4em\relax John Wiley \& Sons, 2012.

\end{thebibliography}

\end{document}